\documentclass[a4paper,10pt]{article}
\usepackage{structure} %

\title{On the Computational Complexity of Linear Discrepancy}
\author{Lily Li\thanks{Department of Computer Science, University of  Toronto, email:\url{xinyuan@cs.toronto.edu}}, Aleksandar Nikolov\thanks{Department of Computer Science, University of Toronto, email:\url{anikolov@cs.toronto.edu}}}

\begin{document}
\maketitle
\begin{abstract}
  Many problems in computer science and applied mathematics require rounding a vector $\vv{w}$ of fractional values lying in the   interval $[0,1]$ to a binary vector $\vv{x}$ so that, for a given matrix $\mm{A}$, $\mm{A}\vv{x}$ is as close to $\mm{A}\vv{w}$ as possible. For example, this problem arises in LP rounding algorithms used to approximate $\class{NP}$-hard optimization problems and in the design of uniformly distributed point sets for numerical integration. For a given matrix $\mm{A}$, the worst-case error over all choices of $\vv{w}$ incurred by the best possible rounding is measured by the linear discrepancy of $\mm{A}$, a quantity studied in discrepancy theory, and introduced by Lovasz, Spencer, and Vesztergombi (EJC, 1986).

  We initiate the study of the computational complexity of linear discrepancy. Our investigation proceeds in two directions: (1) proving hardness results and (2) finding both exact and approximate algorithms to evaluate the linear discrepancy of certain matrices. For (1), we show that linear discrepancy is $\class{NP}$-hard. Thus we do not expect to find an efficient exact algorithm for the general case. Restricting our attention to matrices with a constant number of rows, we present a poly-time exact algorithm for matrices consisting of a single row and matrices with a constant number of rows and entries of bounded magnitude. We also present an exponential-time approximation algorithm for general matrices, and an algorithm that approximates linear discrepancy to within an exponential factor. 
\end{abstract}
\pagebreak

\pagebreak

\section{Introduction}

A number of questions in mathematics and computer science can be
reduced to the following basic rounding question: given a
vector $\vv{w}\in [0,1]^n$, and an $m\times n$ matrix $\mm{A}$, find
an integer vector $\vv{x} \in [0,1]^n$ such that $\mm{A}\vv{x}$ is as
close as possible to $\mm{A}\vv{w}$. For example, many
$\class{NP}$-hard optimization problems can be modeled as an integer
program
\begin{align*}
	\min\quad &\vv{c}^\top \vv{x}\\
	\text{s.t.}\quad &\mm{A}\vv{x} \ge \vv{b}\\
	&\vv{x} \in \{0,1\}^n
\end{align*}
This integer program can be relaxed to a linear program by replacing
the integer variables $\vv{x}\in \{0,1\}^n$ with real-valued variables
$\vv{w} \in [0,1]^n$. A powerful method in approximation algorithms is
to solve this linear programming relaxation to get an optimal
$\vv{w}$, and then round $\vv{w}$ to an integer solution $\vv{x}$
which is feasible (i.e., $\mm{A}\vv{x} \ge \vv{b}$), and has objective
value not much bigger than $\vv{c}^\top \vv{w}$. Often, a useful
intermediate step is to guarantee that $\vv{x}$ is approximately
optimal, i.e., that the coordinates of $\vv{b} - \mm{A}\vv{x}$ are
bounded from above. This approximately feasible solution can then,
hopefully, be turned into a truly feasible one with a small loss in
the objective value. This method was used, for example, by
Rothvoss~\cite{rothvoss2013approximating}, and Rothvoss and
Hoberg~\cite{hoberg2017logarithmic} to give the best known
approximation algorithm for the bin packing problem.

Another example is provided by the problem of constructing uniformly
distributed points, or, more generally, points that are
well-distributed with respect to some measure. Variants of this
problem date back to work by Weyl, van der Corput, van
Aardenne-Ehrenfest, and Roth, and have important applications to such
fields as numerical integration; see the book of
Matou\v{s}ek~\cite{matouvsek1999geometric} for references and an
introduction to the area. In the classical setting, the problem is
to find, for any positive integer $n$, a set of $n$ points $P$ in
$[0,1]^d$, so as to minimize the quantity
\[
  \sup_{R \in \mathcal{R}_d} ||R \cap P| - n \lambda_d(R)|,
\]
where $\mathcal{R}_d$ is the set of all axis-aligned boxes contained
in $[0,1]^d$, and $\lambda_d$ is the Lebesgue measure on $\RR^d$. The quantity
above is known as the (unnormalized) discrepancy of $P$. Note that if
we sample a random point uniformly from $P$, then it would land in $R$
with probability $\frac{|R \cap P|}{n}$; on the other hand, if we
sample a random point uniformly from $[0,1]^d$, then it would land in
$R$ with probability $\lambda_d(R)$. The problem of minimizing the
discrepancy of $P$ is then equivalent to finding a distribution that
is uniform over $n$ points that ``looks the same'' as the continuous
uniform distribution to all boxes $R$.

The discrepancy minimization
problem can be modeled by the rounding problem with which we started
our discussion. To that end, we can discretize the domain $[0,1]^d$ to
a finite set $X$ of size $N$, and let $\mm{A}$ be the incidence matrix
of sets induced by axis-aligned boxes, i.e., each row of $\mm{A}$ is
associated with a box $R$, and equals the indicator vector of
$R\cap X$. If we also let $\vv{w} = \frac{n}{N}\ind{1}$, where
$\ind{1}$ is the all-ones vector, then, for a sufficiently fine
discretization, each coordinate of $\mm{A}\vv{w}$ is a close
approximation of $n\lambda_d(R)$. The problem of finding an $n$-point
set $P$ of minimum discrepancy then becomes essentially equivalent to
minimizing $\|\mm{A}\vv{x} - \mm{A}\vv{w}\|_\infty$ over
$\vv{x} \in \{0,1\}^N$, where $\|\cdot\|_\infty$ is the standard
$\ell_\infty$ norm. In particular, given $\vv{x}$, we can take $P$ to consist of the points in $X$ for which the corresponding coordinate in $\vv{x}$ is
set to $1$. Then, since $[0,1]^d \in \mathcal{R}_d$, we have
$\left||P| - n\right| \le \|\mm{A}\vv{x} - \mm{A}\vv{w}\|_\infty$ and
we can remove or add at most $\|\mm{A}\vv{x} - \mm{A}\vv{w}\|_\infty$
to $P$ to make it exactly of size $n$. The discrepancy of $P$ is then
bounded by $2 \|\mm{A}\vv{x} - \mm{A}\vv{w}\|_\infty$ plus the
additional error incurred by the discretization of $[0,1]^d$. 

These two examples motivate the definition of \emph{linear
  discrepancy}, initially introduced by Lov\'asz, Spencer, and
Vesztergombi~\cite{lovasz1986discrepancy}. The smallest possible error for rounding $\vv{w}$ with
respect to $\mm{A}$ is
\[
	\lindisc(\mm{A}, \vv{w}) = \min_{\vv{x} \in \{0,1\}^{n}} \norm{\mm{A}\left(\vv{w} - \vv{x}\right)}_{\infty}.
\]
This is the linear discrepancy of $\mm{A}$ with respect to
$\vv{w}$. The linear discrepancy of $\mm{A}$ is defined as
the worse case over all $\vv{w}\in [0,1]^n$ i.e.
\begin{equation}
  \label{eq:lindisc}
  \lindisc(\mm{A}) = \max_{\vv{w} \in [0,1]^n} \lindisc(\mm{A}, \vv{w}). %
\end{equation}
It will be useful to consider a maximizer of equation
\eqref{eq:lindisc} i.e. $\mm{w}^* \in [0,1]^n$ such that
$\lindisc(\mm{A}, \mm{w}^*) = \lindisc(\mm{A})$. We call $\vv{w^*}$ a
\emph{deep-hole} of $\mm{A}$. Every $\mm{A}$ has at least one
deep-hole, since linear discrepancy is a continuous function over the
compact set $[0,1]^n$.

The special case of $\lindisc(\mm{A}, \vv{w})$ when $\vv{w} = \frac12 \ind{1}$ is especially
well studied. When $\mm{A}$ is the indicator matrix of a collection
$\family{S}$ of $m$ subsets of a universe $X$, $\lindisc(\mm{A}, \frac12 \ind{1})$ measures
to what extent it is possible to choose a subset $S$ of $X$ that
contains approximately half the elements of each set. This is a
rescaling of the well-known combinatorial discrepancy of $\family{S}$,
defined as
\[\disc(\family{S}) = \min_{\chi:X \to \{-1, +1\}}
  \max_{S \in \family{S}}\left|\sum_{s \in S} \chi(s)\right|
\]
It is straightforward to check that, by a change of variables,
$\disc(\family{S}) = 2\cdot\lindisc(\mm{A}, \frac12 \ind{1})$ where,
again, $\mm{A}$ is the incidence matrix of $\family{S}$. This
definition can be extended to arbitrary matrices $\mm{A}$ as
$\disc(\mm{A}) = 2\cdot\lindisc(\mm{A}, \frac12 \ind{1})$. Combinatorial
discrepancy has been widely studied in combinatorics and computer
science, see~\cite{beck1996discrepancy, chazelle2001discrepancy, matouvsek1999geometric}.

Sometimes $\disc(\mm{A})$ can be small ``by accident", thus it is
useful to define a more robust discrepancy variant.\footnote{Consider
  any matrix $\mm{B} \in \RR^{m \times n}$ and let
  $\mm{A} \in \RR^{m \times 2n}$ be the concatenation of two copies of
  $\mm{B}$ side by side. Regardless of the discrepancy of $\mm{B}$,
  $\disc(\mm{A}) = 0$ since there exists $\vv{x} \in \{-1, 1\}^n$ such
  that $\norm{\mm{A}\vv{x}}_{\infty} = 0$, namely
  \[\vv{x}^{\intercal} = [\underbrace{-1, ..., -1}_{n}, \underbrace{1, ..., 1}_{n}].\]}		
The hereditary discrepancy of $\mm{A}$ is the maximum discrepancy over all sub-matrices, i.e.,
\begin{equation}
  \label{eq:hereditarydiscrepancy}
  \herdisc(\mm{A}) = \max_{\mm{B}} \disc(\mm{B}),
\end{equation}
where $\mm{B}$ ranges over submatrices of $\mm{A}$.

A fundamental theorem by Lov\'asz, Spencer, and Vesztergombi shows that linear
discrepancy can be bounded above by twice the hereditary discrepancy.
\begin{theorem}
  \label{eq:lindiscubstandardherdisc}
  (Lov\'{a}sz et al. 1986, \cite{lovasz1986discrepancy}) $\lindisc(\mm{A}) \leq 2\cdot\herdisc(\mm{A})$.
\end{theorem}
A number of the applications of combinatorial discrepancy use this
basic theorem. In particular, it is common to give an upper bound on
the hereditary discrepancy, and from that deduce an upper bound on the
linear discrepancy. For example, this strategy was used to give
approximation algorithms for bin
packing~\cite{rothvoss2013approximating,hoberg2017logarithmic}, and broadcast scheduling~\cite{BansalKN14},
and to design point sets well distributed with respect to arbitrary
Borel measures~\cite{nikolov17tusnady,transference}. However, linear discrepancy can be
much smaller (by a factor of at least $2^n$) than hereditary
discrepancy,\footnote{Whether this remains true if the matrix $\mm{A}$
  has bounded entries is a tantalizing open question.} so hereditary
discrepancy lower bounds do not translate to linear discrepancy, and,
in general, linear discrepancy lower bounds appear to be
challenging. 
Arguably, a better understanding of linear discrepancy
itself would allow proving more and tighter results, in comparison with going
through hereditary discrepancy. For example, it is likely that new
analytic tools to estimate linear discrepancy would allow progress on
questions in geometric discrepancy theory, as well as questions about the
integrality gaps of linear programming relaxations of important
optimization problems, such as the bin packing problem. 

A sequence of recent works has shed light on the computational
complexity of combinatorial and hereditary discrepancy. It is now
known that combinatorial discrepancy does not allow efficient
approximation algorithms, even in a weak sense (assuming
$\class{P} \neq \class{NP}$)\cite{charikar2011tight}, while hereditary
discrepancy is $\class{NP}$-hard to approximate better than a factor
of two~\cite{austrin20172+varepsilon}, and can be approximated within
poly-logarithmic factors~\cite{matousek18factorization}. Despite being the tool most
directly relevant to many applications of discrepancy, however,
essentially nothing is known about the computational complexity of
linear discrepancy itself. In this paper, we initiate the study of
linear discrepancy from a computational viewpoint, and give both
the first hardness results, as well as the first exact and approximate
algorithms for it.

Before stating our results, it is worth mentioning that linear
discrepancy can also be seen as an analogue of the covering radius in
lattice theory. Let $\lattice \subset \RR^n$ be a lattice,
i.e.~discrete additive subgroup of $\RR^n$, and let us choose
$\vv{b}_1, \ldots, \vv{b}_n$ to be a basis of $\lattice$. Let $\mm{B}$
be a matrix with the $\vv{b_i}$ as its columns. The
covering radius of $\lattice$ in the $\metric{p}$-norm is
defined as
\begin{equation}
  \label{eq:coveringradius}
  \rho(\lattice)
  = \max_{\vv{y} \in \RR^n}\min_{\vv{z} \in     \lattice}\norm{\vv{y} - \vv{z}}_p
  = \max_{\vv{w} \in \RR^n}\min_{\vv{x} \in \ZZ^{n}}\norm{\mm{B} \cdot (\vv{w}-\vv{x})}_p
  = \max_{\vv{w} \in [0,1]^n}\min_{\vv{x} \in \ZZ^{n}}\norm{\mm{B} \cdot (\vv{w}-\vv{x})}_p,
\end{equation}
and is independent of the basis. This definition is equivalent to the
the definition of $\lindisc(\mm{A})$, except that the minimum is over
$\ZZ^n$ rather than $\{0,1\}^n$. Haviv and Regev showed that the
covering radius problem ($\class{CRP}$) in the $\metric{p}$-norm is
$\Pi_2$-hard to approximate within some fixed
constant for all large
enough $p$~\cite{haviv2006hardness}, and Guruswami, Micciancio, and Regev
showed it
can be approximated within a factor of $2^{O(n\log n / \log \log n)}$
for the case of $p =2$~\cite{guruswami2005complexity}.

\subsection{Our Results} 

Let us start with the simple observation that, when $\mm{A}$ is a
single row matrix, deciding $\lindisc(\mm{A}, t\ind{1}) = 0$
is the $\class{NP}$-hard Subset Sum problem with target sum
$t \sum_{j = 1}^n{A_{1,j}}$, and is, therefore,
$\class{NP}$-hard. This does not show, however, that computing
$\lindisc(\mm{A})$ is $\class{NP}$-hard. In this work we show the
following hardness result for linear discrepancy.
\begin{theorem}
  \label{thm:hardness}
  The Linear Discrepancy problem of deciding, given an $m\times n$
  matrix $\mm{A}$ with rational entries, and a rational number $t$,
  whether $\lindisc(\mm{A}) \le t$, is $\class{NP}$-hard and is
  contained in the class $\Pi_2$.
\end{theorem}

We present algorithms for computing linear discrepancy exactly when
the matrix $\mm{A}$ has a constant number of rows. We start with a
result for a single row matrix.
\begin{theorem}
  \label{thm:exact-one-row}
  For any matrix $\mm{A} \in \QQ^{1 \times n}$, $\lindisc(\mm{A})$ can be computed in time $O(n\log n)$. 
\end{theorem}
Note that this stands in contrast to the observation above that
computing $\lindisc(\mm{A}, \vv{w})$ is hard even for a single-row
matrix $\mm{A}$. 
In addition to the theorem above, we also give a corresponding
rounding result, showing that any $\vv{w}\in \QQ^n$ can be efficiently
rounded to within error bounded by the linear discrepancy in the case
of single row matrices. 
\begin{theorem}
  \label{thm:approx-one-row}
  For any matrix $\mm{A} \in \QQ^{1 \times n}$ and any $\vv{w} \in ([0,1]
  \cap \QQ)^n$, we can find an $\vv{x} \in \{0,1\}^{n}$ such that $\norm{\mm{A}(\vv{w} - \vv{x})}_{\infty} \leq \lindisc(\mm{A})$ in time $O(n\log n)$.
\end{theorem}
This result stands in contrast with the hardness of the subset sum
problem, which easily implies that it is $\class{NP}$-hard to round
$\vv{w}$ to within error $\lindisc(\mm{A}, \vv{w})$ even when $\mm{A}$
is a single row matrix.

We can extend Theorem~\ref{thm:exact-one-row} to the case of matrices
with a bounded number of rows, with the additional assumption that the
entries of $\mm{A}$ are bounded. Removing this additional assumption
is a fascinating open question. 
\begin{theorem}
  \label{thm:exact-const-row}
  For any matrix $\mm{A} \in \ZZ^{d \times n}$ where $d$ is some fixed constant and $\max_{i,j}|A_{i,j}| \leq \delta$, $\lindisc(\mm{A})$ can be computed in time $O\left(d(n\delta)^{d^2+d}\right)$.  
\end{theorem}
        
We further present an approximation algorithm for linear discrepancy.
\begin{theorem}
  \label{thm:approx}
  For any matrix $\mm{A} \in \QQ^{m \times n}$, $\lindisc(\mm{A})$ can be
  approximated in polynomial time within a factor of $2^{n+1}$. 
\end{theorem}

\section{Hardness Result}	
	In this section, we show that linear discrepancy ($\class{LDS}$) is
	$\class{NP}$-Hard by reducing from monotone not-all-equal
	3-$\class{SAT}$ ($\class{MNAE3SAT}$)~\cite{gold1978complexity} to
	each. The decision problem version of linear discrepancy we consider
	is defined below.

	\problem{Monotone Not-All-Equal 3-$\class{SAT}$}{MNAE3SAT}{Let $U$ be a collection of variables $\{u_1, ..., u_n\}$ and $\mathcal{C}$ be a 3-$\mathsf{CNF}$ with clauses $\{C_1, ..., C_m\}$ such that $C_i = t_{i,1} \lor t_{i,2} \lor t_{i,3}$ for positive literals $t_{i,j}$.}{Does there exist a truth assignment $\tau: U \rightarrow \{\mathsf{T}, \mathsf{F}\}$ such that $\mathcal{C}$ is satisfied and each clause has at least one true and one false literal?}{}
	
	\problem{Linear Discrepancy}{LDS}{Let $\mm{A} \in \QQ^{m \times n}$ be a matrix and $t \geq 0$ a rational value.}{Is $\lindisc(\mm{A}) \leq t$?}{}

	\subsection{Linear Discrepancy}
	Before we show that linear discrepancy is hard, we will show that the value of $\lindisc(\mm{A})$ can be expressed using a polynomial number of bits in the bit complexity of a matrix for rational matrices. Due to space constraints, the proof can be found in Section~\ref{app:lemma-certificateinpi2}.
	\begin{lemma}
		\label{lem:certificateinpi2}
		For any matrix $\mm{A} \in \QQ^{m \times n}$, there exists a
		deep hole for $\mm{A}$ with bit complexity polynomial in $n$
		and the bit complexity of $\mm{A}$, and, therefore, $\lindisc(\mm{A})$
		can be written in number of bits polynimial in $n$ and the bit
		complexity of $\mm{A}$.
	\end{lemma}
	
	\begin{proof}[Proof of Theorem~\ref{thm:hardness}]
		Note first that the fact that $\class{LDS}$ is contained in $\Pi_2$
		is a straightforward consequence of
		Lemma~\ref{lem:certificateinpi2}: the ``for-all'' quantifier is over
		potential deep holes $\vv{w}\in [0,1]^n$ of the appropriate
		polynomially bounded bit complexity, and the ``exists'' quantifier is
		over $\vv{x} \in \{0,1\}^n$. 
		
		Next we prove hardness.  Let 3-$\mathsf{CNF}$ $\mathcal{C}$ be a
		$\class{MNAE3SAT}$ instance as described above. The corresponding
		$\class{LDS}$ instance will be the incidence matrix
		$\mm{A} \in \{0,1\}^{m \times n}$ of $\mathcal{C}$: column
		$\vv{a}_j$ of $\mm{A}$ corresponds to variable $u_j$ and row
		$\vv{r}_i$ of $\mm{A}$ corresponds to clause $C_i$, and $A_{i,j} =
		1$ if and only if variable $u_j$ appears in clause $C_i$. Let the
		target $t$ in the $\class{LDS}$ problem be $\frac{3}{2} - \epsilon$
		for $\epsilon > 0$ to be determined later.

		Consider first that case that $\mathcal{C}$ is a $\NO$-instance of
		$\class{MNAE3SAT}$ i.e. for every truth assignments $\tau$, there
		exists a clause $C_i$ whose literals all get the same truth
		assignment. Each $\vv{x} \in \{0, 1\}^n$ corresponds to a truth
		assignment. If $x_i = 1$ (resp. $x_i = 0$) then $u_i$ is true
		(resp. $u_i$ is false). Let $C_j$ be the clause whose literals have
		the same truth value. Then
		\[
		\lindisc(\mm{A}) \ge \lindisc(\mm{A}, (1/2)\cdot \ind{1}) \geq
		\left|\vv{r}_j \left(\frac{1}{2} \cdot \ind{1} -
		\vv{x}\right)\right| = \frac{3}{2} > \frac32 - \epsilon,\]
		so $\mm{A}$ is a $\NO$-instance of $\class{LDS}$. 
		
		Consider next the case that $\mathcal{C}$ is a $\YES$-instance of
		$\class{MNAE3SAT}$, and let $\tau$ be a satisfying
		assignment. Suppose $\vv{w}^* \in [0,1]^n$ is a deep-hole of
		$\mm{A}$. If $w_i^* = \frac{1}{2}$ for all $i \in [n]$ then
		\[
		\lindisc(\mm{A}) = \lindisc(\mm{A}, (1/2)\cdot \ind{1})
		=\disc(\mm{A}) \leq \left\|\mm{A}\left(\frac{1}{2} \cdot \ind{1} -
		\vv{x}^{*}\right)\right\|_{\infty} = \frac{1}{2}
		\] 
		since every clause has exactly two elements with the same truth
		value. Thus $\mm{A}$ is a $\YES$-instance of $\class{LDS}$ as long
		as we choose $\epsilon \le 1$. Suppose then that $\vv{w}^* \neq
		\frac12 \ind{1}$, and let $\epsilon$
		be a lower bound on the smallest non-zero gap between $w_i^*$ and $1/2$ i.e. for all
		$w_i^* \neq \frac{1}{2}$,
		\[\left|w_i^* - \frac{1}{2}\right| \geq
		\epsilon.\]
		By Lemma~\ref{lem:certificateinpi2}, which implies that $\vv{w}^*$
		has polynomial bit complexity, we know that we can choose such an
		$\epsilon$ of polynomial bit complexity. 
		We will show that $\lindisc(\mm{A}, \vv{w}^*) \leq \frac{3}{2} - \epsilon$ by constructing a colouring $\vv{x}^*$. Let 
		\[x_i^* = 
		\begin{cases}
		\round(w_i^*) &\mbox{if } w_i^* \neq \frac{1}{2}\\
		\tau(u_i) &\mbox{otherwise}
		\end{cases}\]
		where $\round(w_i^*)$ is $w_i^*$ rounded to the closest integer and $u_i$ is the variable corresponding to column $i$. Let $\vv{r}$ be a row of matrix $\mm{A}$ with non-zero entries in columns $i$, $j$, and $k$. We bound the discrepancy of row $\vv{r}$ based on the number of rounded variables $R_v$ among $\{x_i, x_j, x_k\}$.  
		\begin{enumerate}
			\item[$R_v = 0$:] Since none of the variables are rounded, $w_i^* = w_j^* = w_k^* = \frac{1}{2}$ and 
			\[\left|\vv{r}\left(\vv{x}^*-\vv{w}^*\right)\right| = \left|\left(x_i^* - \frac{1}{2}\right) + \left(x_j^* - \frac{1}{2}\right) + \left(x_k^* - \frac{1}{2}\right)\right| = \frac{1}{2}\]
			since $\tau$ is a satisfying assignment.
			\item[$R_v = 1$:] W.l.o.g assume that that $x_i^*$ is the rounded value and $w_j^* = w_k^* = \frac{1}{2}$. Then  
			\[\left|\vv{r}\left(\vv{x}^* - \vv{w}^*\right)\right| = \left|\left(x_i^* - w_i^*\right) + \left(x_j^* - \frac{1}{2}\right) + \left(x_k^* - \frac{1}{2}\right)\right| \leq \left(\frac{1}{2} - \epsilon\right) + 1 = \frac{3}{2} - \epsilon.\]
			\item[$R_v = 2$:] W.l.o.g assume that $x_i^*$ and $x_j^*$ are the rounded values and $w_k^* = \frac{1}{2}$. Then  
			\[\left|\vv{r}\left(\vv{x}^* - \vv{w}^*\right)\right| = \left|\left(x_i^* - w_i^*\right) + \left(x_j^* - w_j^*\right) + \left(x_k^* - \frac{1}{2}\right)\right| \leq 2 \cdot \left(\frac{1}{2} - \epsilon\right) + \frac{1}{2} = \frac{3}{2} - 2\epsilon.\]
			\item[$R_v = 3$:] All three values are rounded so  
			\[\left|\vv{r}\left(\vv{x}^* - \vv{w}^*\right)\right| = \left|\left(x_i^* - w_i^*\right) + \left(x_j^* - w_j^*\right) + \left(x_k^* - w_k^*\right)\right| \leq 3\cdot\left(\frac{1}{2} - \epsilon\right) = \frac{3}{2} - 3\epsilon.\]
		\end{enumerate}
		Since $\vv{r}$ was an arbitrary row of
	                $\mm{A}$, $\lindisc(\mm{A}) = \lindisc(\mm{A},
	                \vv{w}^*) \leq \frac{3}{2} - \epsilon$ as
	                required. This completes the reduction. 
	\end{proof}

% !TEX root = MAIN-lindisc.tex

\section{Algorithms for Linear Discrepancy}

In the following we consider restrictions and variants of linear
discrepancy for which we are able to give poly-time algorithms. The first subsection considers matrices with a single
row. The second subsection considers matrices
$\mm{A} \in \ZZ^{d \times n}$ with constant $d$ and entry of largest
magnitude $\delta$. In that case, we compute $\lindisc(\mm{A})$ in
time $O\left(d(2n\delta)^{d^2}\right)$. The third subsection presents
a poly-time $2^n$ approximation to $\lindisc(\mm{A})$ for
$\mm{A} \in \QQ^{m \times n}$.

\subsection{Linear Discrepancy of a Row Matrix} 	
\label{sssection:onerow}
We begin by developing some intuition for the linear discrepancy of a
one-row matrix, $\mm{A} = [a_1, ..., a_n]$. For now, let us make the
simplifying assumption that the entries of $\mm{A}$ are non-negative
and sorted in decreasing order. Define the \emph{subset sums} of
$\mm{A}$ to be the multi-set
$\family{S}(\mm{A}) = \{s_1, ..., s_{2^n}\}$ where each
$s_i = \mm{A}\vv{x}$ for exactly one $\vv{x} \in \{0,1\}^n$. Enumerate
the element of $\family{S}(\mm{A})$ in non-decreasing order, i.e.
$s_i \leq s_{i+1}$. If $\ell_{A} = 2\cdot\lindisc(\mm{A})$, then
$\ell_A$ is the width of the largest gap between consecutive entries
in $\family{S}(\mm{A})$.
	
Suppose $\mm{A}_i = [a_1, ..., a_i]$. Let us consider how
$\family{S}(\mm{A}_i)$ and $\lindisc(A_i)$ change for the first couple
of values of $i$. Clearly, $\family{S}(\mm{A}_1) = [0, a_1]$ and
$\lindisc(\mm{A}_1) = \frac{a_1}{2}$. $\family{S}(\mm{A}_2)$ is the
disjoint union of $\family{S}(\mm{A}_1)$ and
$\family{S}(\mm{A}_1)$ shifted to the right by $a_2$. Since
$a_1 \geq a_2$, $\family{S}(\mm{A}_2) = [0, a_2, a_1, a_1 + a_2]$
where the largest gap is of size $\max(a_2, a_1 - a_2)$. See Figure
\ref{fig:onerowalgexp}. In general, the entries of
$\family{S}(\mm{A}_i)$ consists of two copies of
$\family{S}(\mm{A}_{i-1})$ with one shifted to the right by $a_i$. The
gaps in $\family{S}(\mm{A}_i)$ are gaps previously in
$\family{S}(\mm{A}_{i-1})$ or between an element of
$\family{S}(\mm{A}_{i-1})$ and one in
$\{a_i + s: s \in \family{S}(\mm{A}_{i-1})\}$.
	
	\fig{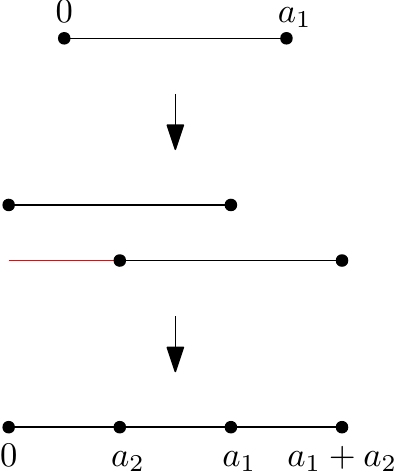}{1}{Obtaining $\family{S}(\mm{A}_2)$ from $\family{S}(\mm{A}_1)$ when $a_1 \geq a_2$.}
	
	A similar structure occurs for general matrices with real valued entries with two caveats: (1) the previous interval is shifted left or right depending on the sign of the current entry (negative and positive respectively) and (2) the smallest entry of $\family{S}(\mm{A})$ is not zero but the sum of the negative entries in $\mm{A}$.
	
	\begin{lemma}
          \label{lem:onerowsortinghelps}
          Suppose $\mm{A}_{k-1} = [a_1, ..., a_{k-1}]$ with entries in $\RR$ and $|a_i| \geq |a_{i+1}|$. Let the largest gap in $\family{S}(\mm{A}_{k-1})$ be of size $\ell_{k-1}$. Then, for $\mm{A}_{k} = [a_1, ..., a_{k-1}, a_{k}]$ where $|a_{k}| \leq |a_{i}|$ for all $i \in [k-1]$, the largest gap in $\family{S}(\mm{A}_k)$ is of size $\max(|a_k|, \ell_{k-1} - |a_k|)$.
        \end{lemma}
	\begin{proof}
          Again, it is important to remember that the entries of $\family{S}(\mm{A}_{k})$ are exactly those in $\family{S}(\mm{A}_{k-1})$ along with those in $\{a_k + s: s \in \family{S}(\mm{A}_{k-1})\}$. Let $\ell = \max(|a_k|, \ell_{k-1} - |a_k|)$.
		
          We first show that $2\cdot\lindisc(\mm{A}_{k}) \leq \ell$ by
          showing that gaps between consecutive entries in
          $\family{S}(\mm{A}_{k})$ have size at most $\ell$. If
          $(s_{j}, s_{j+1})$ is a consecutive pair in
          $\family{S}(\mm{A}_{k-1})$ such that $s_{j+1} - s_j > \ell$,
          then $s_j$ and $s_{j+1}$ are no longer consecutive in
          $\family{S}(\mm{A}_{k})$, since
          $s_{j} \leq s_{j} + a_k \leq s_{j+1}$ if $a_k > 0$ and
          $s_{j} \leq s_{j+1} + a_k \leq s_{j+1}$ if $a_k <
          0$. See Figure \ref{fig:polytimealgrowmatrixlemma1}.
          Then, the gap given by any such pair gets split into gaps of
          size at most $\max\{|a_k|, s_{j+1} - s_j - |a_k|\} \le \ell$,
          where the inequality holds because $s_{j+1} - s_j \le \ell_{k-1}$.
          It follows that the size of each gap in
          $\family{S}(\mm{A}_k)$ is at most $\ell$.

          \fig{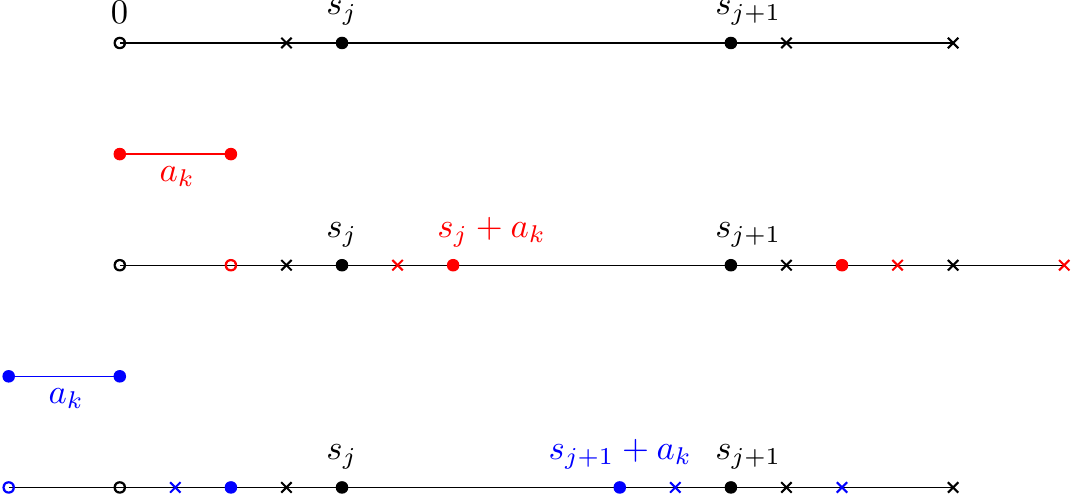}{1}{All consecutive pairs in $\family{S}(\mm{A}_{k-1})$ of size greater than $|a_k|$ will be divided into two or more consecutive pairs in $\family{S}(\mm{A}_k)$. The red interval indicates what happens when $a_k > 0$. The blue interval indicates what happens when $a_k < 0$.} 
		
          Next we will show that
          $2\cdot\lindisc(\mm{A}_{k}) \geq \ell$ by producing a pair
          of consecutive entries in $\family{S}(\mm{A}_{k})$ which
          achieves gap $\ell$. Suppose $\ell = |a_{k}|$. Recall that
          $s_0$ is the smallest subset sum of all entries in
          $\mm{A}_{k}$, which equals the sum of all negative entries
          in $\mm{A_k}$. Then it is easy to check that $s_1$ equals
          $s_0 + |a_k|$, where we recall that $a_k$ is the entry in
          $\mm{A}_k$ with minimum absolute value. Therefore, 
          $(s_0, s_0 + |a_{k}|)$ is a consecutive pair in
          $\family{S}(\mm{A}_k)$. This means that if $\ell = |a_k|$,
          then we are done, as we have produced a pair with gap $\ell$.
		
          When $\ell = \ell_{k-1} - |a_k| > |a_k|$, we split our analysis into two cases: (1) $a_k > 0$ and (2) $a_k < 0$. 
		
		In the former case, let $(s_{j^*}, s_{j^*+1})$ be a consecutive pair in $\family{S}(\mm{A}_{k-1})$ that achieves gap $\ell_{k-1}$ and suppose, towards a contradiction, that $s_{j^*} + a_k$ and $s_{j^*+1}$ do not appear consecutively in $\family{S}(\mm{A}_k)$. Then there must be some $s \in \family{S}(\mm{A}_k)$ such that $s_{j^*} + a_k < s < s_{j^*+1}$. Note that $s$ cannot be an element of $\family{S}(\mm{A}_{k-1})$ since $s_{j^*}$ and $s_{j^*+1}$ are consecutive in $\family{S}(\mm{A}_{k-1})$, so $s - a_k$ must be an element of $\family{S}(\mm{A}_{k-1})$. However, since $s > s_{j^*} + a_k$, we have $s - a_k > s_{j^*}$. This is a contradiction since $s_{j^*}$ and $s_{j^*+1}$ are consecutive entries in $\family{S}(\mm{A}_{k-1})$. See Figure \ref{fig:polytimealgrowmatrix}. Thus $(s_{j^*} + a_k, s_{j^*+1})$ must be a consecutive pair in $\family{S}(\mm{A}_k)$.
		
		\fig{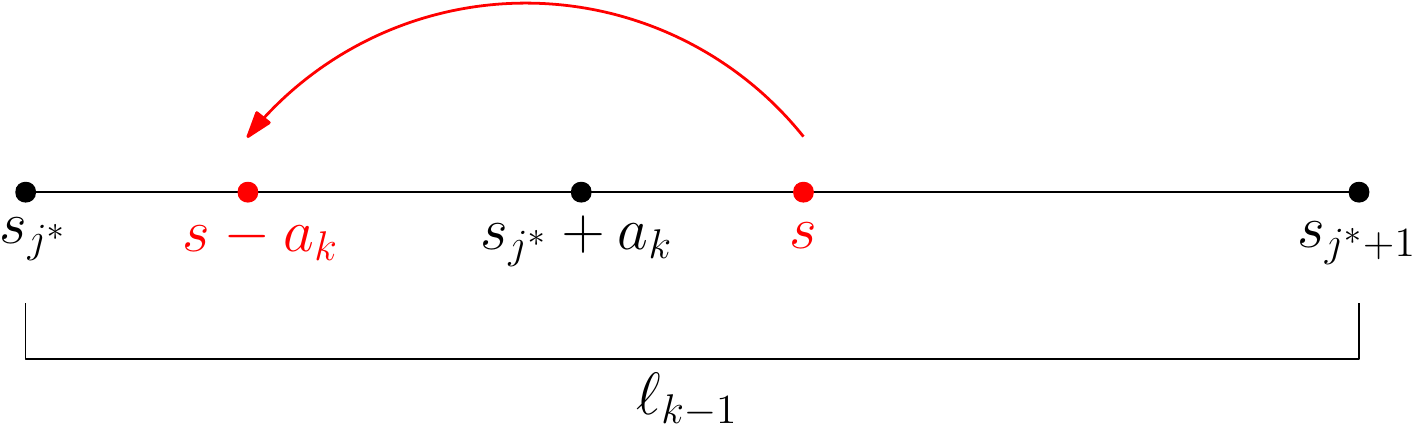}{0.6}{Suppose $a_k < \ell_{k-1} - a_k$ and there exists $s \in \family{S}(\mm{A}_{k})$ such that $s_{j^*} + a_k < s < s_{j^*+1}$.}
		
		The latter case, when $a_k < 0$, is similar. Again there exists a pair of consecutive entries $(s_{j^*}, s_{j^*+1})$ in $\family{S}(\mm{A}_{k-1})$ which achieves gap $\ell_{k-1}$. Suppose, towards contradiction, that $s_{j^*}$ and $s_{j^*+1} - |a_k|$ do not appear consecutively in $\family{S}(\mm{A}_k)$. Then there must be some $s \in \family{S}(\mm{A}_k)$ such that $s_{j^*} < s < s_{j^*+1} - |a_k|$. Again, $s$ cannot be an element of $\family{S}(\mm{A}_{k-1})$ since $s_{j^*}$ and $s_{j^*+1}$ are consecutive in $\family{S}(\mm{A}_{k-1})$, so $s + |a_k|$ must be an element of $\family{S}(\mm{A}_{k-1})$. However since $s < s_{j^*+1} - |a_k|$, we have $s + |a_k| < s_{j^* + 1}$. This is a contradiction since $s_{j^*}$ and $s_{j^*+1}$ are consecutive entries in $\family{S}(\mm{A}_{k-1})$.
	\end{proof}

        Lemma~\ref{lem:onerowsortinghelps} has the following curious corollary.
	\begin{corollary}
		\label{cor:onerowtakemagnitudes}
		Let $\mm{A} = [a_1, ..., a_n]$ and $\mm{A}' = [|a_1|, ..., |a_n|]$. Then $\lindisc(\mm{A}) = \lindisc(\mm{A}')$.
	\end{corollary}
	
	Lemma \ref{lem:onerowsortinghelps} and Corollary \ref{cor:onerowtakemagnitudes} suggest an algorithm: replace the entries of $\mm{A}$ by their magnitudes. Sort $\mm{A}$. Consider each entry in turn and update the largest gap accordingly. See Algorithm \ref{alg:lindisc}.
	
	\begin{proof}[Proof of Theorem \ref{thm:exact-one-row}]
		By Corollary $\ref{cor:onerowtakemagnitudes}$ it is sufficient to consider row matrices with non-negative entries. Suppose that $\mm{A} = [a_1, ..., a_n]$ is such a matrix with entries sorted in decreasing order. Algorithm \ref{alg:lindisc} correctly outputs the linear discrepancy for matrices with a single entry. Let $\mm{A}_i = [a_1, ..., a_i]$. Lemma \ref{lem:onerowsortinghelps} gives us a recursive method for computing the largest gap in $\family{S}(\mm{A}_{i+1})$ from the largest gap in $\family{S}(\mm{A}_i)$. Since $\lindisc(\mm{A})$ is half the size of the largest gap in $\family{S}(\mm{A})$, Algorithm \ref{alg:lindisc} computes $\lindisc(\mm{A})$ as required.
	\end{proof} 
	
	\begin{algorithm}
		\DontPrintSemicolon % Some LaTeX compilers require you to use \dontprintsemicolon instead 
		\KwIn{Matrix $\mm{A} \in \QQ^{1\times n}$.}
		\KwOut{$\lindisc(\mm{A})$.}
		\For{$i$ from $1$ to $n$}{
			$\mm{A}[i] \leftarrow |a_i|$\;
		}
		sort $\mm{A}$ in decreasing order\;
		$\ell \leftarrow a_1$\;
		\For{$i$ from $2$ to $n$}{
			$\ell \leftarrow \max(a_i, \ell - a_i)$\;
		}
		\Return{$\frac{\ell}{2}$}\;
		\caption{Linear discrepancy of row matrix.}
		\label{alg:lindisc}
	\end{algorithm}
	
	Thus, for any row matrix $\mm{A}$ with $n$ elements, we can find $\lindisc(\mm{A})$ in time $O(n \log n)$.  
	
	% \begin{corollary}
	% 	If $\mm{A}$ is an arithmetic progression of the form $[a_0, a_0 + k, \cdots, a_{0} + (n-1)k]$, then 
	% 	\[\lindisc(A) = \frac{\max(a_0, k)}{2}.\]
		
	% 	Similarly, if $\mm{A}$ is a geometric progression of the form $[a_0r^{0}, \cdots, a_{0}r^{n-1}]$, then 
	% 	\[\lindisc(A) = \begin{cases}
	% 	\frac{a_0}{2} &\mbox{if } r \leq 1,\\
	% 	a_0\left(r^{n-1} - \frac{\left(r^{(n-1)}-1\right)}{(r-1)}\right) &\mbox{otherwise}.
	% 	\end{cases}\]
	% \end{corollary}
	
	\subsubsection{One Row Linear Discrepancy Rounding}	
	Let $\lindisc(\mm{A}) = \ell$. By the definition of linear discrepancy, for every $\vv{w} \in [0,1]^{n}$ there exists an $\vv{x} \in \{0,1\}^{n}$ such that $\norm{\mm{A}(\vv{w} - \vv{x})}_{\infty} \leq \ell$. In-fact, if $\vv{w}$ is not a deep-hole, there exists an $\vv{x}$ which satisfies $\norm{\mm{A}(\vv{w} - \vv{x})}_{\infty} < \ell$. However it is not obvious that finding such an $\vv{x}$ can be done efficiently i.e. in polynomial time with respect to the bit complexity of $\mm{A}$ and $n$. By reducing from the subset-sum problem, we observe that it is difficult to compute $\lindisc(\mm{A}, \vv{w})$ let alone find an $\mm{x}$ which minimizes $\norm{\mm{A}(\vv{w} - \vv{x})}_{\infty} \leq \ell$.
	
	\begin{proof}[Proof of Theorem \ref{thm:approx-one-row}]
		To begin, let $\mm{A} = [a_1, ..., a_n]$ for positive $a_i$ in non-increasing order. We will consider $\mm{A}$ with arbitrary entries at the end. Let $w = \mm{A}\vv{w}$. As before, let $\family{S}(\mm{A}) = [s_{0}, ..., s_{2^{n}-1}]$ be the subset-sums of $\mm{A}$ where each $s_i = \mm{A}\vv{x}$ for an $\vv{x} \in \{0,1\}^n$ and $s_i \leq s_{i+1}$ for all $i$. Recall that $2\cdot\lindisc(\mm{A})$ is the largest gap between any two consecutive entries in $\family{S}(\mm{A})$. Our algorithm will find a pair of subset sums containing $w$. If we can show that the size of the interval between these two subset sums is no more than the gap between some two consecutive entries in $\family{S}(\mm{A})$, then the closest subset sum to $w$ among these two will be within $\lindisc(\mm{A})$ of $w$.
		
		Just as in Algorithm \ref{alg:lindisc}, we refine the interval between two subset sums containing $w$ by incrementally adding the entries of $\mm{A}$ in decreasing order. Initially our interval is $g_0 = [0, \sum_{i = 1}^{n} a_i]$. We maintain the invariants: (1) $w \in g_i$ for all $i$, and (2) the end-points of $g_i$ are subset sums.
		
		Suppose $w \in g_i = [u, v]$ and we are considering $a_i$. If $u + a_{i} > w$ then set $v \leftarrow \min(v, u + a_i)$. Otherwise let $u \leftarrow u + a_i$. Algorithm \ref{alg:lindiscvariant} computes this interval and the associated vectors $\vv{u}$ and $\vv{v}$ representing its endpoints.
		
		Consider the values of $u$ and $v$ at the end of the algorithm. We claim that the final interval $[u,v]$ is at most the width of some gap between two consecutive terms in $\family{S}(\mm{A})$, the array of all subset sums of $\mm{A}$. Notice $u = a_1u_1 + \cdots + a_nu_n$ where $\vv{u} = [u_1, ..., u_n]$ is an endpoint of the interval once Algorithm \ref{alg:lindiscvariant} completes.  
		
		We partition $\vv{u}$ into maximal blocks where all entries in the same block have the same value i.e. $[u_1, u_2, ..., u_{\ell_1}], ..., [u_{\ell_r + 1}, u_{\ell_r+2}, ..., u_{n}]$ such that $u_{\ell_i + 1} = u_{\ell_i + 2} = \cdots = u_{\ell_{i+1}}$ for $i = 0, 1, ..., r-1$ where $\ell_0 = 0$.
		
		We claim that Algorithm \ref{alg:lindiscvariant} outputs an interval containing $w$ whose width is at most the distance between some two consecutive entries in $\family{S}(\mm{A})$. The proof is by induction on $r$, the number of blocks. In the base case, $r = 1$ and there is only one block. Thus $u = 0$ or $u = \sum a_i$. In the case where $u = 0$, we must have $a_i > w$ for all $i \in [n]$, and $v = a_n$. Thus $w \in [0, a_n]$ with consecutive elements $0$ and $a_n$ of $\family{S}(\mm{A})$. In the latter case when $u = \sum a_i$, we can output $\vv{w}$ since it is already a subset sum.
		
		Suppose next that the claim holds for all matrices where the algorithm outputs a vector $\vv{u}$ with $k$ blocks, and we will show that it still holds for a matrix $\mm{A}$ whose output $\vv{u}$ has $k+1$ blocks. Let $\vv{u}' = [u_1, ..., u_{\ell_{k+1}}]$ and $\vv{v}' = [v_1, ..., v_{\ell_{k+1}}]$ be the final vectors after running the algorithm on $\mm{A}' = [a_1, ..., a_{\ell_{k+1}}]$. Further let $u' = \sum_{i = 1}^{\ell_{k+1}} a_iu_i$ and $v' = \sum_{i = 1}^{\ell_{k+1}} a_iv_i$. By the induction hypothesis, the width of $[u', v']$ is at most the distance between some two consecutive elements in the list of subset sums of $\mm{A}'$. The last block of $\vv{u}$ is $[u_{\ell_{k+1} + 1}, ..., u_{n}]$. The entries of this block are either all zeros or all ones. Consider each case in-turn. 
		
		First suppose $u_{\ell_{k+1}+1} = \cdots = u_{n} = 0$. Since none of the $a_i$ for $i = \ell_{k+1} + 1, ..., n$ were added to $u$, it must be the case that $u' + a_i > w$ for all such $i$. Thus the interval $[u, v] = [u', \min\left(v', u' + a_n\right)]$ has width at most $a_n$. Since $0$ and $a_n$ are consecutive in $\family{S}(\mm{A})$, as $|a_n|$ is the entry with the smallest magnitude in $\mm{A}$, the output interval satisfies our requirements.  
		
		Next suppose $u_{\ell_{k+1}+1} = \cdots = u_{n} = 1$. It must be the case that $u = u' + a_{\ell_{k+1} + 1} + \cdots + a_n \leq w$. Observe that $a_{\ell_{k+1}}$ is in the $k$\textsuperscript{th} block and so $u_{\ell_{k+1}} = 0$. Let $[u'', v'']$ be our interval after processing the $k-1$\textsuperscript{st} block i.e. $u'' = \sum_{i = 1}^{\ell_k} a_iu_i$ and $v'' = \sum_{i = 1}^{\ell_k} a_iv_i$. Notice that since none of the entries in the $k$\textsuperscript{th} block were added to $u''$, we must have $u'' + a_{i} > w$ for all $i = \ell_{k} + 1, ..., \ell_{k+1}$. In such cases, we always update $v'' \leftarrow \min(v'', u'' + a_{i})$ after each such $i$, thus the interval $[u', v']$ has width at most $a_{\ell_{k+1}}$. Thus it suffices to show that $a_{\ell_{k+1} + 1} + \cdots + a_n$ and $a_{\ell_{k+1}}$ are consecutive in $\family{S}(\mm{A})$. First note that $a_{\ell_{k+1} + 1} + \cdots + a_n \le a_{\ell_{k+1}}$ since $u' + a_{\ell_{k+1} + 1} + \cdots + a_n \le w \le v' \le u' + a_{\ell_{k+1}}$. The two subset sums then are also consecutive, since $a_i > a_{\ell_{k+1}}$ for all $i < \ell_{k+1}$. 
		
		Now consider the case where $\mm{A}$ can have both positive and negative entries. Without loss of generality we can assume that none of the entries are zero. Let ${A}_{-} = \{a_i \in \mm{A}: a_i < 0\}$ and ${A}_{+} = \{a_i \in \mm{A}: a_i > 0\}$. It suffices to set $u_0 = \sum_{a \in {A}_{-}}a$ and $v_0= \sum_{a \in {A}_{+}} a$ and let $\vv{u}$ and $\vv{v}$ be the indicator vectors of $\mm{A}_{-}$ and $\mm{A}_{+}$ respectively. The remainder of the algorithm is identical except that the matrix should be sorted in decreasing order of \emph{magnitude} and every time an element $a_i \in \mm{A}_{-}$ is added to $u$, its entry in $\vv{u}$ should be set to zero. 
	\end{proof} 
	
	\begin{algorithm}
		\DontPrintSemicolon % Some LaTeX compilers require you to use \dontprintsemicolon instead 
		\KwIn{A vector $\vv{w} \in [0,1]^n$ and a row matrix $\mm{A} = [a_1, ..., a_n]$ of positive integers sorted in increasing order.}
		\KwOut{A vector $\vv{x} \in \{0,1\}^{n}$ such that $\norm{\mm{A}(\vv{w} - \vv{x})}_{\infty} \leq \lindisc(\mm{A})$.}
		$\mm{A} \leftarrow \op{sort-decreasing}(\mm{A})$\;
		$\vv{u} \leftarrow \op{zeros}(n)$\;
		$\vv{v} \leftarrow \op{ones}(n)$\;
		$w \leftarrow \mm{A}\vv{w}$, $u \leftarrow \mm{A}\vv{u}$, $v \leftarrow \mm{A}\vv{v}$\;
		\Return $\vv{v}$ if $w == v$\;
		\For{$k = 1..n$}{
			\If{$u + a_k > w$}{
				$v \leftarrow \min\left(v, u + a_k\right)$\;
				\If{$v == u + a_k$}{
					$\vv{v} \leftarrow \op{copy}(\vv{u})$\;
					$\vv{v}[k] \leftarrow 1$\;
				}
			} \Else{
				$u \leftarrow u + a_{k}$\;
				$\vv{u}[k] \leftarrow 1$\;
			}
		}
		\Return{$\vv{u}$ if $u$ is closer to $w$ else $\vv{v}$}\;
		\caption{Finding a close subset sum to $\mm{A}\vv{w}$.}
		\label{alg:lindiscvariant}
	\end{algorithm}

	\subsection{Constant Rows with Bounded Matrix Entries}
	Let $\mm{A} \in \ZZ^{d \times n}$ with $\max_{i,j} |A_{i,j}| \leq \delta$. Let $Z = \mm{A}[0,1]^{d}$ be the zonotope of $\mm{A}$ and let $T = [-n\delta, n\delta]^{d} \cap \ZZ^{d}$ be the set of all integer lattice points of $Z$. The following algorithm computes $\lindisc(\mm{A})$ in polynomial time with respect to $n$ for fixed $d$ and $\delta$. The algorithm makes use of Lemma~\ref{lem:lec-in-higher-dimensions}, which is proved in the Appendix.

	\begin{proof}[Proof of Theorem \ref{thm:exact-const-row}]
		
	For every one of the $(2n\delta + 1)^{d}$ integral points $\vv{b} \in T$, compute whether $\mm{A}\vv{x} = \vv{b}$ for some $\vv{x}\in \{0,1\}^n$ using dynamic programming. This procedure generalizes dynamic programming algorithms for knapsack and subset sum and will be outlined in the following. Let $\vv{a}_1, ..., \vv{a}_n$ be the columns of $\mm{A}$. Construct a matrix $\mm{M}$ with dimensions $[-n\delta, n\delta]^{d} \times n$. Cell $(\vv{v}, i)$ of $\mm{M}$ contains the indicator $[\mm{M}(\vv{v}-\vv{a}_i, i-1) \lor \mm{M}(\vv{v},i-1)]$; this corresponds to a linear combination of the first $i-1$ columns of $\mm{A}$ which adds up to $\vv{v}-\vv{a}_i$ or a linear combination of the first $i-1$ columns which adds up to $\vv{v}$. The first column of $\mm{M}$ is the indicator vector for $\{\vv{a}_1\}$. Computing the entries of $\mm{M}$ takes time $O(2n\delta)^{d+1}$. $\mm{M}(\vv{b},n)$ indicates the feasibility of $\mm{A}\vv{x} = \vv{b}$. Computing this for all $\vv{b}$ takes time $O(2n\delta)^{d+1}$. Let $S \subseteq T$ be the set of points $\vv{b}$ in $Z$ such that $\mm{A}\vv{x}=\vv{b}$ for some $\vv{x}\in \{0,1\}^n$, and set $|S| = N$.
	
	Apply Lemma~\ref{lem:lec-in-higher-dimensions} to the points of $S$ in $\metric{\infty}$-norm. The output is some radius $r$ and point $\vv{x}^*$ such that the $\ell_\infty$-ball centered at $\vv{x}^{*}$ with radius $r$ is the largest such ball with center inside the convex hull of $S$ not containing any points of $S$. Note that $r$ is in-fact the linear discrepancy of $\mm{A}$. Since $r$ and $\vv{x}^*$ can be computed in time $O(N^d)$, $\lindisc(\mm{A})$ can be computed in time $O(2n\delta)^{d^2 + d}$. 
	\end{proof}

	\subsection{Poly-time Approximation Algorithm}
	Next, we prove Theorem~\ref{thm:approx}, presenting a
        $2^n$-approximation algorithm for linear discrepancy. Recall
        that $\round(\vv{w})$ is the function which rounds each
        coordinate of $\vv{w}$ to its nearest integer (with ties
        broken arbitrarily). Let the $p$-to-$q$ operator norms of a matrix $\mm{A}$ be:
    \[\norm{\mm{A}}_{p \rightarrow q} = \max_{\vv{x}\in \RR^n\setminus\{0\}} \frac{\norm{\mm{A}\vv{x}}_{q}}{\norm{\vv{x}}_{p}}.\]
    Note that 
	\[\lindisc(\mm{A}) \leq \max_{\vv{w} \in [0,1]^n} \norm{\mm{A}(\mm{w} - \round(\mm{w}))}_{\infty} \le \frac{1}{2} \max_{\vv{z} \in [-1,1]^n}\norm{\mm{A}\vv{z}}_{\infty} = \frac{1}{2}\norm{\mm{A}}_{\infty \rightarrow \infty}.\]
	
	To bound $\lindisc(\mm{A})$ from below, we show that $\norm{\mm{A}}_{\infty \rightarrow \infty} \leq 2^{n+1} \cdot \lindisc(\mm{A})$. This completes the proof of the theorem, since $\norm{\mm{A}}_{\infty \to \infty}$ equals the largest $\ell_1$ norm of any row of $\mm{A}$, and can be computed in polynomial time.

        Let us try to interpret the statement $\norm{\mm{A}}_{\infty \rightarrow \infty} \leq 2^{n+1} \cdot \lindisc(\mm{A})$. Note that $\norm{\mm{A}\vv{z}}_{\infty}$ is equal to the Minkowski $\polytope{P}$-norm $\norm{\vv{z}}_{\polytope{P}}$ for $\polytope{P} = \{\vv{x}: \norm{\mm{A}\vv{x}}_{\infty} \leq 1\}$ i.e. $\norm{\vv{z}}_{\polytope{P}} = \inf\{t \geq 0: \vv{z} \in t\polytope{P}\}$ so
	\[\norm{\mm{A}}_{\infty \rightarrow \infty} = \max_{\vv{z} \in [-1, 1]^n} \norm{\mm{A}\vv{z}}_{\infty} = \max_{\vv{z} \in [-1, 1]^n} \norm{\vv{z}}_{\polytope{P}}.\]
	By interpreting $\vv{z}$ as the difference of two vectors $\vv{x}, \vv{x}' \in [0,1]^n$ we have that 
	\[\norm{\mm{A}}_{\infty \rightarrow \infty} = \max_{\vv{z} \in [-1,1]^n} \norm{\vv{z}}_{\polytope{P}} = \max_{\vv{x}, \vv{x}' \in [0,1]^n} \norm{\vv{x} - \vv{x}'}_{\polytope{P}}.\]
	It is an easy, and well-known fact that $\lindisc(\mm{A})$ is the smallest $t$ such that $[0,1]^n \subseteq \bigcup_{\vv{x}\in \{0,1\}^n}(\vv{x} + \polytope{P})$; see  \cite{matouvsek1999geometric}. We then just need to show that the diameter of the unit hyper-cube with respect to the Minkowski $\polytope{P}$-norm is no more than this scale-factor $t$ times $O(2^n)$. 
	We prove the following more general statement.
	\begin{lemma}
		\label{lem:approx-algorithm}
		Let $\polytope{K}$ be a convex symmetric polytope and $S \subset \RR^n$ be convex. Suppose there exist $N$ elements $x_1, ..., x_N \in S$ such that 
		\[S \subseteq \bigcup_{x_i} x_i + t\polytope{K}.\]
		Then $ \max_{x, x' \in S} \norm{x - x'}_{\polytope{K}} \leq 2tN$.
	\end{lemma}
	\begin{proof}
		Fix any two points $x$ and $x'$ in $S$. Let $\polytope{P}_i$ be the polytope $x_i + t\polytope{K}$. Since $S$ is convex, the line segment $\lambda x + (1 - \lambda)x'$ for $\lambda \in [0,1]$ is in $S$. Therefore $\lambda x + (1 - \lambda)x'$ intersects a sequence of polytopes $\polytope{P}_{k_1}, ..., \polytope{P}_{k_r}$ with centres $x_{k_1}, ..., x_{k_r}$, such that any two consequtive polytopes in the sequence intersect. Since the polytopes are convex, we can assume that they appear in the sequence at most once, so $r \le N$. By the triangle inequality we have
		\begin{align*}
			\norm{x - x'}_K 
			&= \norm{(x - x_{k_1}) + (x_{k_1} - x_{k_2}) + \cdots + (x_{k_r} - x')}_{K}\\ 
			&\leq \norm{x - x_{k_1}}_{K} + \norm{x_{k_1} - x_{k_2}}_{K} + \cdots + \norm{x_{k_r} - x'}_{K}\\
			&\leq t + 2t(N-1) + t = 2tN
		\end{align*}
		where the last inequality follows as $x \in \polytope{P}_{k_1}$, $x' \in \polytope{P}_{k_r}$, and $\norm{x_{k_i} - x_{k_{i+1}}}_K \leq 2t$.
		\begin{comment}
		Just for completeness, we show that $\norm{x_{k_i} - x_{k_{i+1}}}_K \leq 2t$. Since $P_{k_i}$ and $P_{k_{i+1}}$ are consecutive in the sequence of polytopes that $\lambda x + (1 - \lambda) x'$ intersect, $P_{k_i} \cap P_{k_i} \neq \emptyset$. Let $y$ be a point in their intersection. Then 
		\[\norm{x_{k_i} - x_{k_{i+1}}}_K \leq \norm{x_{k_i} - y}_{K} + \norm{y - x_{k_{i+1}}}_K \leq 2t.\]
		\end{comment}
	\end{proof}
	
	\begin{proof}[Proof of Theorem \ref{thm:approx}]
		In Lemma \ref{lem:approx-algorithm}, set $\polytope{K}$ to be the parallelepiped defined by $\mm{A}$, $S = [0,1]^n$, $t = \lindisc(\mm{A})$, and $\{x_1, ..., x_N\} = \{0,1\}^{n}$. 
	\end{proof}

\section{Open Problems}
Because of the similarity between the closest vector problem and
linear discrepancy, we suspect that linear discrepancy is also
$\Pi_2$-complete, and the hardness result of Theorem
\ref{thm:hardness} is, in this sense, not tight. Further, Haviv and
Regev also showed that $\class{CRP}$ is $\Pi_2$-hard to approximate to
with-in a factor of $\frac{3}{2}$, and we conjecture that a similar hardness
of approximation result should hold for linear discrepancy.

We suspect that the algorithm used to prove Theorem
\ref{thm:exact-one-row} can be generalized to matrices $\mm{A} \in
\QQ^{d \times n}$ with running time $\tilde{O}(n^{d})$. This would be
a substantial improvement on the $O\left(d(n\delta)^{d^2}\right)$
running time algorithm used to prove Theorem
\ref{thm:exact-const-row}, and would be independent of the magnitude of the largest entry of $\mm{A}$. 

It is also interesting to extend the largest empty ball algorithm from
Lemma~\ref{lem:lec-in-higher-dimensions} to other $\ell_p$ norms, or
even arbitrary norms, given appropriate access to the norm
ball. Currently, this seems rather difficult as  Voronoi diagrams with
respect to the $\metric{p}$-norm for $p \in (2, \infty)$ are poorly behaved. For the standard
$\metric{2}$-norm Voronoi diagram in $\RR^{d}$, it is the case that $d
+ 1$ affinely independent vertices are equidistant to exactly one
point. This is no longer the case even in $\RR^{3}$ for
$\metric{4}$-norm \cite{icking1995convex}. In particular, there exists
a set of four vertices such that the intersection of their pair-wise
bisectors has size three. The situation is even worse for general strictly convex norms. There exists such norms where the pair-wise bisectors of a set of four points in $\RR^{3}$ can have arbitrarily many intersections.

We currently also have no evidence that the approximation factor in
Theorem~\ref{thm:approx} is tight. One possibility is that there
exists an approximation preserving reduction from the closest vector
problem in lattices to linear discrepancy. This would show
that one cannot expect a significant improvement to
Theorem~\ref{thm:approx} without also improving the best polynomial
time approximation to the covering radius, which is currently also
exponential in the dimension $n$. On the other hand, we also
conjecture that the approximation factor in Theorem~\ref{thm:approx}
can be taken to be a function of $\min\{m, n\}$, or even of the rank
of the matrix $\mm{A}$.

\printbibliography

\appendix
\section{Appendix}
\subsection{Bit Complexity of Linear Discrepancy}
\label{app:lemma-certificateinpi2}
\begin{proof}[Proof of Lemma \ref{lem:certificateinpi2}.]
	Let  $\mm{r}_i$ for $i \in [m]$ be the rows of $\mm{A}$, $\lindisc(\mm{A}) = \lambda_A$, and $\vv{w}^*$ be a deep-hole of $\mm{A}$. For every $\vv{x} \in \{0,1\}^{n}$ there exists an $i \in [m]$ and $\sigma \in \{-1, 1\}$ such that $\sigma\vv{r}_i(\vv{w}^* - \vv{x}) \geq \lambda_A$. Let $\vv{b}_{x} = \sigma\mathbf{r}_i$ and consider the following linear program over the variables $\vv{w} \in \RR^n$ and $\lambda \in \RR$:
	\begin{align*}
	\mbox{Maximize: } &\lambda\\
	\mbox{Subject to: } &\mathbf{b_x}(\mathbf{w}-\mathbf{x}) \geq \lambda &\mbox{for all } \mathbf{x} \in \{0,1\}^n\\
	&\mathbf{0} \leq \mathbf{w} \leq \mathbf{1}
	\end{align*}
	
	Let $\lambda^*$ be the optimum value of this
	linear program. First note that $\lambda_A
	\leq \lambda^*$ since $(\vv{w}^*,\lambda)$
	satisfies the constraints. Next we show that
	$\lambda_A \geq \lambda^*$. Suppose, towards contradiction, that $\lambda_A < \lambda^*$. Then there exists $\vv{w}' \in [0,1]^n$ such that 
	\[\|\mm{A}(\vv{w}' - \vv{x})\|_\infty \ge \vv{b}_x(\vv{w}' - \vv{x}) \ge \lambda^*>\lambda_A\] 
	for every $\vv{x} \in \{0,1\}^{n}$. Since $\lambda_A = \lindisc(\mm{A})$, we cannot have $\lindisc(\mm{A}, \mathbf{w}') > \lambda_A$. Thus $\lambda^* = \lindisc(\mm{A})$. Since this LP has $n$ variables, the number of bits required to express the linear discrepancy and some deep-hole $\vv{w}^*$ of $\mm{A}$ are polynomial in $n$ and the bit complexity of the largest entry of $\mm{A}$ \cite{schrijver1998theory}. 
\end{proof} 
	
\subsection{Largest Empty Ball Problem}
\label{app:largest-empty-ball}
Let $V$ be a set of $n$ points in the plane and let $\ch(V)$ denote the convex hull of $V$. The largest empty circle problem, denoted LEC, takes $V$ and outputs both a radius $r$ and point $\vv{x}^* \in \ch(V)$ such that the circle centered at $\vv{x}^*$ with radius $r$ is the largest empty circle not containing any point of $V$. We generalize this problem to other norms and to higher dimensions as follows: $V$ is a set of $n$ points in $\RR^d$, and the goal is to compute a point $\vv{x}^*$ in $\ch(V)$ such that $\vv{x}^* + rB$ does not contain any point of $V$, where $B$ is the unit ball of either the $\ell_2^d$ or the $\ell_\infty^d$ norm. In the following we present an algorithm which solves this largest empty ball (LEB) problem.

\begin{lemma}{\textup{(LEC in Higher Dimensions.)}}
	\label{lem:lec-in-higher-dimensions}
	Let $V$ be a set of $n$ points in $\RR^{d}$ for some fixed constant $d$. The LEB of $V$, in both $\metric{2}$- and $\metric{\infty}$-norms, can be computed in time $O(n^{d})$.
\end{lemma}
\begin{proof}
	We use the following terminology. Define a face $F$ of the
        Voronoi diagram $\mathrm{vd}(V)$ of $V$ to be a subset of
        $\RR^d$ such that, for some $S\subseteq V$, and every $\vv{x}
        \in F$, $S$ are the points in $V$ closest to $\vv{x}$. In
        particular, this means that any $\vv{x}\in F$ is equidistant
        from all points in $S$. 
	
	The algorithm of Toussaint \cite{toussaint1983computing}
        computes the LEB of $n$ points $V$ in the plane with respect to the $\metric{2}$-norm as follows, 
	\begin{enumerate}
		\item Compute $\mathrm{vd}(V)$. Note that
                  $\mathrm{vd}(V)$ is the union of Voronoi faces of
                  dimension $k$, the set of which we denote $\mathrm{vd}_{k}(V)$, over all $k = 0, ..., d-1$. 
		\item Compute the convex hull of $V$, denoted $\ch(V)$. Let $h$ be the number of facets of $\ch(V)$. 
		\item Preprocess the points of $\ch(V)$ so that queries of the form ``Is a point $x$ in $\ch(V)$?'' can be answered in time $O(\log h)$. For every $v \in \mathrm{vd}_{0}(V)$, determine if $v \in \ch(V)$. Let $C_1 = \{v \in \mathrm{vd}_{0}(V): v \in \ch(V)\}$.
		\item Determine the intersection points of faces in $\mathrm{vd}_{k}(V)$ with faces of $\ch(V)$ of co-dimension $k$, for pairs of such faces that intersect at a unique point. Let $C_2$ be the set of all such intersection points.
		\item For all points $v \in C_1 \cup C_2$, find the largest empty circle centered at $v$. Output a $v$ which maximizes this radius.
	\end{enumerate}
	We find the analogue of each step for points in $\RR^{d}$ with
        respect to the $\metric{2}$-norm, and then adapt the algorithm
        to the $\metric{\infty}$-norm.
	
	In the following let $N = n^{\ceil{d/2}}$. The complexity, i.e. total number of faces of every dimension, of the $\metric{2}$-Voronoi diagram in $\RR^{d}$ for fixed $d$ is $O(N)$ and can be computed in time $O(N + n \log n)$ by a classic result of Chazelle \cite{chazelle1993optimal}. The complexity of $\ch(V)$ is $O(N)$ and can also be computed in time $O(N + n\log n)$.
	
	To determine the set $C_1$ of Voronoi intersection points inside the convex hull, we let $\mathcal{H}$ be the set of bounding hyperplanes of $\ch(V)$. Assume, without loss of generality, that $\ch(V)$ contains the origin, and, for each $H \in \mathcal{H}$, let $H^-$ be the half-space with $H$ as its boundary containing the origin. Then $\ch(V) = \bigcap_{H \in \mathcal{H}} H^-$. We simply test, for each Voronoi intersection point $\vv{v}$, whether $\vv{v} \in H^-$ for each $H\in \mathcal{H}$, in total time $O(N)$. Since there are at most $O(N)$ Voronoi intersection points, we can find $C_1$ in time $O(N^2)$. 
	
	To determine the set $C_2$ of all unique intersection points
        of $k$-faces of $\mathrm{vd}_k(V)$ and faces of $\ch(V)$ of
        co-dimension $k$ will require solving several linear
        systems. Note that the points in each face $F$ in
        $\mathrm{vd}_{k}(V)$ satisfy $d-k$ equality constraints
        $\ip{\vv{a}_1}{\vv{x}} = b_1, \ip{\vv{a}_k}{\vv{x}} = b_k$ for
        linearly independent vectors
        $\vv{a}_1, ... \vv{a}_k \in \RR^{d}$. Similarly, the points in
        each face of co-dimension $k$ of $\ch(V)$ satisfy $k$ linearly
        independent equality constraints. Since there are at most
        $O(2^dN) = O(N)$ faces of $\ch(V)$, there are at most that
        many faces of $\ch(V)$ of co-dimension $k$. We can then go
        over all Voronoi faces $F$ of dimension $k$, and all faces $G$
        of $\ch(V)$ of co-dimension $k$, and solve the corresponding
        system of $(d-k) + k = d$ linear equations. If the system has
        a unique solution, we check if that solution is in $F\cap G$,
        and, if so, we add it to $C_2$. Thus, for constant $d$, the
        size of $C_2$ and the time to compute it are bounded bounded
        above by $O(N \cdot 2^d N) = O(N^2)$.
	
	In total there are at most $O(N + N^2)$ points in $C_1 \cup C_2$ which can be computed in time $O(N^2)$. Thus solving the largest empty ball problem in dimension $d$ for constant $d$ takes time $O(n^{d})$. 
	
	Next we consider the largest empty ball problem in
        $\metric{\infty}$-norm. The convex hull remains the same, so
        we just have to consider the Voronoi diagram with respect to
        the $\metric{\infty}$-norm. Again, constructing the Voronoi
        diagram can be done in expected time
        $O(n^{\ceil{d/2}}\log^{d-1}n)$ using the randomized algorithm
        of Boissonnat et al.
        \cite{DBLP:journals/dcg/BoissonnatSTY98}. Next we consider the
        number of intersections between the Voronoi diagram and the
        convex hull. First note that Voronoi diagrams with respect to
        the $\metric{\infty}$-norm need not consist of only
        hyperplanes and their intersections. Indeed, in $\RR^{d}$, for
        two points with the same $y$-coordinate, there exists regions
        with affine dimension two which are equidistant to both
        points. To remedy this, we assume that no two points in $V$
        have the same $i$-th coordinate, for any $i \in [d]$. This is
        without loss of generality, by perturbing the points in $V$
        slightly. It remains to consider the complexity of each
        bisector in $\metric{\infty}$-norm. By Claim
        \ref{claim:bound-facet-linf-bisector}, in constant dimension
        $d$, each such bisector can have at most $O(d^2)$
        facets. Therefore, the complexity of any face of the Voronoi
        diagram, being the intersection of at most $d$ bisectors, is
        bounded by a function of $d$. Thus the bounds of the
        $\metric{2}$-norm algorithm still hold, up to constant factors
        that depend on $d$. 
\end{proof}

\begin{claim}{\textup{(Bound on Number of Facets of $\metric{\infty}$ Bisectors.)}}
	\label{claim:bound-facet-linf-bisector}
	Let $\vv{u}, \vv{v} \in \RR^{d}$ be such that assume that $u_i \neq v_i$ for all $i \in [d]$. Then the bisector $\{\vv{x}: \|\vv{x} - \vv{u}\|_\infty = \|\vv{x} - \vv{v}\|_\infty\}$ has at most $O(d^{2})$ facets.
\end{claim}
\begin{proof}
  Let $\vv{x}$ be a point in the bisector at $\metric{\infty}$ distance $r$ from $\vv{u}$ and $\vv{v}$. Pick coordinates $i$ and $j$ and signs $\sigma$ and $\tau$ in $\{-1, +1\}$ such that
  \begin{equation}
    \label{eq:bisector-facet}
    \sigma_i (x_i - u_i) = \tau_j ({x}_j - v_j) = r.
  \end{equation}
  Moreover, let us make this choice so that either $i \neq j$ or
  $\sigma_i \neq \tau_i$. This is always possible, since, otherwise,
  the assumption on $\vv{u}$ and $\vv{v}$ is violated. Then,
  \eqref{eq:bisector-facet} defines a hyperplane in $\RR^d$, namely
  $H_{i, j, \sigma, \tau} =\{\vv{x}: \sigma_i x_i - \tau_j x_j =
  \sigma_i u_i - \tau_j v_j\}$. Note that there are at most
  ${2d\choose 2} \in O(d^2)$ such hyperplanes, and each $\vv{x}$ in
  the bisector lies in at at least one of them. Moreover, a point $\vv{x}$
  in $H_{i,j,\sigma, \tau}$ lies in the bisector if and only if it
  satisfies the inequalities 
  \begin{align*}
    |x_k - u_k| &\le \sigma_i (x_i - u_i) \ \ \forall k \in [d],\\
    |x_k - v_k| &\le \tau_j ({x}_j - v_j) \ \ \forall k \in [d].\\
  \end{align*}
  Thus, the bisector is the union of $(d-1)$-dimensional convex
  polyhedra, one per each of the $O(d^2)$ hyperplanes $H_{i,j,\sigma,
    \tau}$. 
\end{proof}

\end{document}